\documentclass{article}
\usepackage[utf8]{inputenc}
\providecommand{\keywords}[1]
{
  \small	
  \textbf{\textit{Keywords}} #1
}
\title{On the duality between homological quantum codes of a hypermap and its dual hypermap}
\author{Zihan Lei}
\date{
School of Mathematical Sciences USTC, China\\
\today}
\usepackage{amsmath}
\usepackage{amssymb}
\usepackage{amsthm}
\usepackage{appendix}
\newtheorem{theorem}{Proposition}
\usepackage{wrapfig}
\usepackage[export]{adjustbox}
\usepackage{tikz}
\usetikzlibrary{commutative-diagrams}
\DeclareMathOperator{\im}{im}

\usepackage{graphicx}
\graphicspath{ {images/} }

\begin{document}

\maketitle

\begin{abstract}
 From a given topological hypermap \(H\), we define two related hypermaps \(H^\triangle\) and \(H^\nabla\) as complements of the ordinary dual hypermap \(H^*\) along with the concepts of their edge hypermap quantum codes \(\mathcal{C}^\triangle\) and \(\mathcal{C}^\nabla\). We then show that, when the sets of special darts are naturally corresponded, the duality between the ordinary hypermap quantum code \(\mathcal{C}\) from \(H\) and the one \(\mathcal{C}^*\) from \(H^*\) can be greatly simplified to the duality between \(\mathcal{C}^\triangle\) and \(\mathcal{C}^\nabla\).
\end{abstract}


\maketitle

\section{Introduction}

Martin \cite{Martin} has described a new kind of stabilizer codes called hypermap-homology quantum codes (Abbreviated as `hypermap codes' or `hypermap quantum codes' in the following part.), which are homological quantum codes constructed from the \(\mathbb{Z}_2\)-chain complexes of hypermap homology \cite{Homology}, a kind of homology rooted in the geometry of topological hypermaps.  Having a topological hypermap \(H=(\Sigma,\Gamma)\) at hand , Pradeep \cite{Pradeep} then showed that there is a surface code \cite{PHD,Projective} which is equivalent to the hypermap code of \(H\) and whose cell-structure is obtained by adding curves (We will call them Pradeep's curves later.) to the orientable surface \(\Sigma\) of $H$. 

Pradeep's work means that the parameters of hypermap codes cannot be superior to a kind of well known stabilizer codes, i.e, the surface codes. On the other hand, hypermap codes can be regarded as a new approach to obtain surfaces codes and topological hypermaps themselves also have rich connections with geometry and topology. Therefore, we believe that further exploration of the connections between quantum error correction codes and the math of hypermaps is worthy. As a preliminary attempt, we try to study the relations between quantum code $\mathcal{C}$ constructed from a hypermap using Martin’s method and the code $\mathcal{C}^*$ from the dual hypermap \cite{Martin} by the same method. During the exploring process, we will mainly use pradeep's curves as a visual aid.

The main contributions of the paper are as follows: (1) In \cite{Martin}, Martin has related the classical cohomology from the bipartite graph \(\Gamma^*\) of the dual hypermap $H^*=(\Sigma_{op},\Gamma^*)$ to the hypermap cohomology of \(H\) in his Proposition 4.21. In this article, we find that there is actually a very simple geometrical relation between the bipartite graph \(\Gamma^*\) and the Praddeep's curves of $H$. (2) We defined another kind of dual hypermap \(H^\triangle\) with respect to $H$, which we call $\triangle$-dual, along with the concept of it's contrary map $H^\nabla$, and we give a proof that $H^\nabla$ is exactly the  $\triangle$-dual of the dual hypermap $H^*$ in the sense of strong isomorphism, i.e, \(H^\nabla=(H^*)^\triangle\). Meanwhile, we proposed a new kind of homological quantum codes that can be constructed from a given topological hypermaps, which we call the edge hypermaps codes, while the Martin's type of hypermap codes are renamed the face hypermap codes. Then we find that for a given topological hypermap, the face hypermap code is equivalent to the edge hypermap code of its $\triangle$-duals, thus, we can transform the pair $(\mathcal{C},\mathcal{C}^*)$  of face hypermap codes of $H$ and $H^*$ to the equal pair $(\mathcal{C}^\triangle,\mathcal{C}^\nabla)$ of edge hypermap codes of $H^\triangle$ and $H^\nabla$. This transformation illuminates the relations between $\mathcal{C}$ and $\mathcal{C}^*$ in the following aspects: (i) as graphs, the bipartite graphs (hypergraphs) of $H^\triangle$ and $H^\nabla$ coincide, the only difference is that their sets of hyperedges and hypervertices are opposite. (ii) the sets of special darts of $\mathcal{C}^\triangle$ and $\mathcal{C}^\nabla$ coincide. (iii) the combinatorial hypermap of $H^\nabla$ is just that of \(H^\triangle\) with $\alpha$ and $\sigma$ interchanged.

The article is structured as follows. In section 2, we review the necessary background  of homological quantum codes, conbinatorial and topological hypermaps, then hypermap quantum codes and their Pradeep's curves. In section 3.1, we give a geometric duality between Pradeep's surface code and the dual hypermap. In section 3.2, which is the main part, we define $\triangle$-duals, contrary maps and edge hypermap quantum codes, and we use them to uncover the relations between $\mathcal{C}$ and $\mathcal{C}^*$.

\section{Background}
\subsection{Homological quantum codes}
In this section, we briefly review the part of the construction of homological quantum codes that we need , assuming the reader is familiar with the basis of CSS stabilizer codes and surface codes.  A chain complex is a sequence of vector spaces \(V_i\) with linear morphisms \(d_i\) in between satisfying \(d_{i}\circ d_{i+1}=0\). In the context of homological quantum codes, we consider only the shortest kind, which consists of three vector spaces and two morphisms as follow:
\begin{figure}[ht]
\centering
\begin{tikzpicture}[codi]
\obj{V_{i+1} & V_i & V_{i-1} \\};
\mor V_{i+1} {d_{i+1}}:-> V_i ;
\mor V_i {d_i}:-> {V_{i-1}} ;
\end{tikzpicture}
\end{figure}

 \noindent where all vector spaces are limited to \(\mathbb{Z}_2\)-vector spaces. To construct an error correcting code, we choose a basis for each vector space. Denoting \([d_i]\) the matrix of \(d_i\) for these bases, we have \([d_{i}\circ d_{i+1}]=[d_i][d_{i+1}]=0\), which means that we can use \(A=\)
$\big(\begin{smallmatrix}
  H_X & 0\\
  0 & H_Z
\end{smallmatrix}\big)$ 
as the binary check matrix for a Calderbank, Shor, and Steane (CSS)  code with \(H_X=[d_i]\) and \(H_Z=[d_{i+1}]^T\). If we denote \(C_X\) and \(C_Z\) the kernel of the matrices \(H_X\), \(H_Z\), then by the standard theory of CSS code, the number of logical qubits should be
\begin{equation}
    k=n-\dim(C_X^\perp)-\dim(C_Z^\perp)
\end{equation}
where n denotes the dimension of the central vector space \(V_i\). On the other hand, by definition of homology groups, that is, the quotient spaces \(H_i=\ker d_{i}/\im d_{i+1}\), we also have \(\dim H_i=\dim(\ker{d_i})-\dim(\im{d_{i+1}})=n-\dim{C_X^\perp}-\dim{C_Z^\perp}=k\), which indicates that the numbers of logical qubits of these codes depend only on dimensions of the central homology groups of their chain complexes. We call these codes the homological quantum code.
\subsection{Hypermaps and hypermap codes}

We're going to review some basic concepts of hypermaps, their \(\mathbb{Z}_2\)-homology and hypermap-homology quantum codes.

There are two kinds of hypermaps, combinatorial hypermaps and topological hypermaps, which are mutually transformable. A \emph{combinatorial hypermap} is easy to define, it's a pair of elements \((\alpha, \sigma)\) of the group \(S_n\) of all permutations on \(B_n=\{1,2,...,n\}\) under composition with \(<\sigma, \alpha>\) transitive on \(B_n\). Here, `transitive' means that every two elements of \(B_n\) can be transformed to each other by an element of the subgroup \(<\sigma, \alpha>\). To define topological hypermaps, we need to  define the concept of a hypergraph at first, here, we use the Walsh representation of hypergraph as its definition, which is equivalent to the original definition of hypergraph. A (connected) \emph{hypergraph} is simply a (connected) bipartite graph. However, the edges of the original bipartite graph are renamed the darts of the hypergraph, while we call the vertices of the bipartite graph that are naturally divided into two separate subsets \(V\) and \(E\) the `\emph{(hyper)virtices}' and `\emph{(hyper)edges}' of the corresponding hypergraph. Then, a (oriented) topological hypermap $H=(\Sigma, \Gamma)$ is an embedding of a hypergraph $\Gamma$ into a oriented compact surface \(\Sigma\), with $\Sigma\setminus{\Gamma}$ consists of finite many separated open sets \(e_\alpha\), each \(e_\alpha\) homeomorphics to $\mathbf{R}^2$ 
\begin{equation}
    \Sigma\setminus{\Gamma}=\bigsqcup_{\alpha}^{} e_\alpha 
\end{equation}
, and is called a face or a (open) cell.

To transform a topological hypermap to a combinatorial one, note first that at every sites (we use the world `sites' to replace `vertices' of the original bipartite graph, which consist of both vertices and edges of the hypergraph.), the hypergraph is star-shaped locally---there is one site and the darts that incident to it. Also note that because the surface is oriented, there exist a global normal vector field with respect to the orintation. By seeing from the top of the normal vector at each site, we can orient the darts clockwise or counterclockwise. We denote B the set of all darts of the hypergraph, which is assumed to be label by the number set \(B_n\) with \(n=|B|\) , then define a permutation \(\alpha\in{S_B}\) which takes a dart to the next one that clockwise around the edge to which it incident ,(with the normal vector at that edge pointing to us.)  and another permutation \(\sigma\in{S_B}\) which takes a dart to the next one that counterclockwise around the vertex to which it incident . Now we have the pair \((\alpha, \sigma)\), and the transitivity of \(<\sigma, \alpha>\) is easily seem from the connectivity of the hypergraph. When drawing pictures, we usually use a small circle to represent a vertex, a small square to represent an edge, and let the normal vectors point out from the papper. (Of course, we can only draw part of the whole hypermap this way.) Also, we write the lable (a number in \(B_n\)) of a dart inside the face that counterclockwise incident to the dart with respect to to normal vector at its edge, and say that the dart belongs to this face, this way, one dart can only belongs to one face. 

Figure \ref{fig:1} shows an example of a picture of a hypermap, as we can see, the darts labeled 3, 11 and 7 belong to face \(f_2\), while the darts labeled 5 and 6 belong to the face \(f_1\). The darts 4 doesn't belong to any of these two faces but another face which is not explicitly labeled here. Also, we have \(\sigma(4)=5\), \(\alpha(6)=7\) and \(\alpha^{-1}(4)=7\).\par
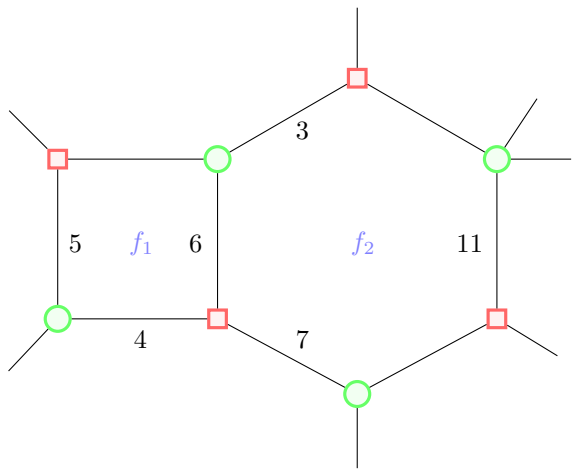
\begin{figure}
\centering
\begin{tikzpicture}[
roundnode/.style={circle, draw=green!60, fill=green!5, very thick, minimum size=2mm},
squarednode/.style={rectangle, draw=red!60, fill=red!5, very thick, minimum size=2mm},
fakenode/.style={rectangle, draw=orange!0, fill=blue!0, very thick, minimum size=2mm},
]
\node[squarednode]      (edge 1){} ;
\node[fakenode]        (site 1)       [below=0.8cm of edge 1]{};
\node[roundnode]        (vertex 3)       [below=2.8cm of site 1]{};
\node[roundnode]      (vertex 2)       [left=1.532cm of site 1] {};
\node[roundnode]      (vertex 1)       [right=1.532cm of site 1] {};
\node[squarednode]      (edge 3)       [below=1.8cm of vertex 2] {};
\node[squarednode]      (edge 2)       [below=1.8cm of vertex 1] {};
\node[roundnode]      (vertex 4)       [left=1.8cm of edge 3] {};
\node[squarednode]      (edge 4)       [left=1.8cm of vertex 2] {};
\node[fakenode]        (site 3)       [above=0.507cm of edge 4]{};
\node[fakenode]        (site 4)       [left=0.507cm of site 3 ]{};
\node[fakenode]        (site 6)       [below=0.507cm of vertex 4]{};
\node[fakenode]        (site 5)       [left=0.507cm of site 6]{};
\node[fakenode]        (site 2)       [above=0.8cm of edge 1]{};
\node[fakenode]        (site 7)       [below=0.8cm of vertex 3]{};
\node[fakenode]        (site 11)       [right=0.3cm of vertex 1]{};
\node[fakenode]        (site 12)       [right=0.8cm of vertex 1]{};
\node[fakenode]        (site 13)       [above=0.666cm of site 11]{};
\node[fakenode]        (site 9)       [right=0.666cm of edge 2]{};
\node[fakenode]        (site 8)       [below=0.3cm of site 9]{};
\node[fakenode]        (site 10)       [above=0.1cm of site 1]{};
\node[fakenode]        (site 14)       [left=0.366cm of site 10]{3};
\node[fakenode]        (site 15)       [below=2.27cm of site 14]{7};
\node[fakenode]        (site 16)       [below=0.8cm of vertex 2]{};
\node[fakenode]        (site 17)       [right=2.9cm of site 16]{11};
\node[fakenode]        (site 18)       [left=3.1cm of site 17]{6};
\node[fakenode]        (site 19)       [left=1.15cm of site 18]{5};
\node[fakenode]        (site 20)       [left=1.7cm of site 15]{4};
\node[fakenode]        (site 21)       [right=0.35cm of site 19]{\textcolor{blue!50}{\( f_1\)}};
\node[fakenode]        (site 22)       [left=0.8cm of site 17]{\textcolor{blue!50}{\(f_2\)}};

\draw[] (edge 1) -- (site 2);
\draw[] (edge 2) -- (site 8);
\draw[] (edge 4) -- (site 4);
\draw[] (vertex 1) -- (site 13);
\draw[] (vertex 1) -- (site 12);
\draw[] (vertex 3) -- (site 7);
\draw[] (vertex 4) -- (site 5);
\draw[] (vertex 1) -- (edge 2);
\draw[] (vertex 1) -- (edge 1);
\draw[] (vertex 2) -- (edge 1);
\draw[] (vertex 2) -- (edge 3);
\draw[] (vertex 3) -- (edge 3);
\draw[] (vertex 3) -- (edge 2);
\draw[] (vertex 2) -- (edge 4);
\draw[] (vertex 4) -- (edge 4);
\draw[] (vertex 4) -- (edge 3);
\end{tikzpicture}
\caption{Part of a hypermap}
\label{fig:1}
\end{figure}

Here, we make some remarks about the notations that we will encounter. For each number \(i\in B_n\), there is exactly one dart $\omega_i\in B$ that has \(i\) as its label. When a permutation \(h\in S_B\) acts on $\omega_i$, we denote it by $h\omega_i$, when a permutation \(g\in S_{B_n}\) acts on $i$, we denote it by $g(i)$. However, permutations like $\alpha$, $\sigma$ defined above can be seen to act on both the darts and their labels, in which situations, we treat expressions like \(\omega_{\alpha^{-1}(i)}\) and $\alpha^{-1}\omega_i$ the same by abusing of notations.
This would not cause any problem when there is only one topological hypermap. However, when there are more than one hypermaps with a common lebel set $B_n$, which is case in the third part of this article, one should be more careful.
For expressions like \(\alpha^{-1}\sigma\), we take the convention in \cite{Martin}, that is, acting from left to right. 
\begin{figure}
\centering
\begin{tikzpicture}[
roundnode/.style={circle, draw=black!100, fill=green!0, minimum size=2mm},
squarednode/.style={rectangle, draw=black!100, fill=red!0, minimum size=2mm},
fakenode/.style={rectangle, draw=orange!0, fill=blue!0, very thick, minimum size=2mm},
]
\node[fakenode]        (site 1)       {};
\node[squarednode]      (edge 1)       [left=1.8 of site 1] {};
\node[squarednode]      (edge 5)       [right=1.8 of site 1] {};
\node[fakenode]        (site 2)      [left=1.214 of site 1] {};
\node[roundnode]      (vertex 2)       [above=1.214cm of site 2] {};
\node[squarednode]      (edge 3)       [above=1.8 of site 1] {};
\node[fakenode]        (site a^1)      [right=0.45 of edge 3 ]{\(\alpha^{-1}(i_1)\)};
\node[roundnode]      (vertex 4)       [right=2.628 of vertex 2] {};
\node[roundnode]      (vertex 8)       [below=2.628 of vertex 2] {};
\node[squarednode]      (edge 7)       [below=1.8 of site 1] {};
\node[roundnode]      (vertex 6)       [below=2.628 of vertex 4] {};

\node[fakenode]        (site i_1)      [right=0.45 of vertex 2] {\(i_1\)};
\draw[] (vertex 2) -- (edge 3);
\node[fakenode]        (site i_2)      [below=0.45 of vertex 4]{\(i_2\)};
\draw[] (vertex 4) -- (edge 5);
\node[fakenode]        (site i_3)      [left=0.45 of vertex 6]{\(i_3\)};
\draw[] (vertex 6) -- (edge 7);
\node[fakenode]        (site i_4)      [above=0.45 of vertex 8]{\(i_4\)};
\draw[] (vertex 4) -- (edge 3);
\node[fakenode]        (site a^3)      [left=0.45 of edge 7 ]{\(\alpha^{-1}(i_3)\)};
\draw[] (vertex 8) -- (edge 7);
\node[fakenode]        (site op)      [below=0.05 of site i_4 ]{};
\node[fakenode]        (site a^2)      [right=3.2 of site op  ]{\(\alpha^{-1}(i_2)\)};
\draw[] (vertex 6) -- (edge 5);
\node[fakenode]        (site a_4)      [left=3.55 of site i_2 ]{};
\node[fakenode]        (site a^4)      [above=0.01 of site a_4 ]{\(\alpha^{-1}(i_4)\)};
\draw[] (vertex 2) -- (edge 1);
\draw[] (vertex 8) -- (edge 1);
\node[roundnode]      (vertex 2)       [above=1.214cm of site 2] {};
\end{tikzpicture}
\caption{A face for \(s=4\)} 
\label{fig:2}
\end{figure}
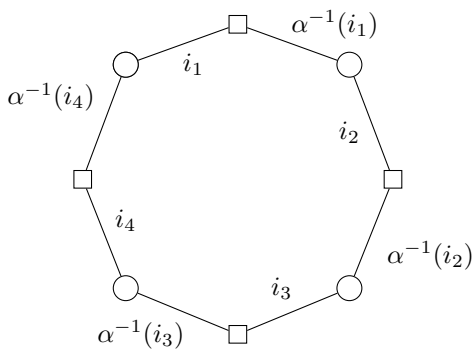
Now, because any face of $H$ can be uniquely determined by the darts that belongs to it, which forms an orbit of \(<\alpha^{-1}\sigma>\), (See Figure \ref{fig:2}) the faces are in one-to-one correspondence with the orbits of \(<\alpha^{-1}\sigma>\). For the same reason, edges (vertices) are in one-to-one correspondence with the orbits of \(<\alpha^{-1}>\) ($<\sigma>$). If we denote $O_{\owns{i}}$ the orbits (of some given subgroup of $S_B$) that $i$ belongs to, then we understand, for example, that $v_{\owns{i}}$ is the vertex on which dart $i$ incident.

On the contrary, we can also construct Topological hypermap from a conbinatorial one \cite{Martin}, which is not needed here. See Appendix A for this along with an example.

From a topological hypermap, we have four \(\mathbb{Z}_2\)-vector spaces \(\mathcal{V}\), \(\mathcal{E}\), \(\mathcal{F}\), \(\mathcal{W}\) with bases the vertices, edges, faces and darts of the hypermap. Define \(d_2(f)=\sum_{i\in{f}}^{} \omega_i\), which map a face to the sum of the darts that belong to it. We also define \(d_1(w_i)= v_{\owns{i}} + v_{\owns{\alpha^{-1}(i)}}\) , and \(\iota(e)=\sum_{i\in{e}}\omega_i\) which maps an edge to all darts \(\omega_i\) that incident on it. Based on these definitions, we get the linear map \(d_2: \mathcal{F} \rightarrow \mathcal{W}\), \(d_1: \mathcal{W} \rightarrow \mathcal{V}\), and \(\iota : \mathcal{E} \rightarrow \mathcal{W}\) by extending linearly to the whole vector spaces. Now, it is straightforward to check that \(d_1\circ{d_2}=0\), \(d_1\circ\iota=0\). The second identity tells us that we can define a map \(\partial_1\) from the quotient space \(\mathcal{W}/\iota(\mathcal{E})\) to \(\mathcal{V}\) by \(\partial_1(\omega_i + \iota(\mathcal{E}))=d_1(\omega_i)\) without ambiguity. We also define \(p : \mathcal{W} \rightarrow \mathcal{W}/\iota(\mathcal{E})\) the natural projection and \(\partial_2=p\circ{d_2}\). All these maps are given in the commutative diagram shown in Figure \ref{jiaohuantu}. Now, we have constructed a short chain complex \(\partial_1\circ\partial_2=0\), form which we get what we called hypermap-homology.
\begin{wrapfigure}{r}{0.4\textwidth}
\centering
\begin{tikzpicture}[
squarednode/.style={rectangle, draw=black!0, fill=green!0, very thick, minimum size=2mm},
]
\node[squarednode]      (site 1)        {\(\mathcal{W}\)};
\node[squarednode]      (site 2)       [right=1.6 of site 1] {\(\mathcal{V}\)};
\node[squarednode]      (site 3)       [left=1.6 of site 1] {\(\mathcal{F}\)};
\node[squarednode]      (site 4)       [below=1 of site 1] {\(\mathcal{W}/\iota(\mathcal{E})\)};
\node[squarednode]      (fake l)       [left=0.3 of site 4]{};
\node[squarednode]      (fake r)       [right=0.3 of site 4]{};
\node[squarednode]      (p1)       [above=0.1 of fake l]{\(\partial_2\)};
\draw[->] (site 3) -- (site 4);
\node[squarednode]      (p2)       [above=0.1 of fake r]{\(\partial_1\)};
\draw[->] (site 4) -- (site 2);
\node[squarednode]      (fake l+)       [left=0.4 of site 1]{};
\node[squarednode]      (fake r+)       [right=0.4 of site 1]{};
\node[squarednode]      (d1)       [above=-0.2 of fake l+]{\(d_2\)};
\node[squarednode]      (d2)       [above=-0.2 of fake r+]{\(d_1\)};
\draw[->] (site 1) -- (site 2);
\draw[->] (site 3) -- (site 1);
\node[squarednode]      (center)       [below=0.3 of site 1]{};
\node[squarednode]      (p)       [right=-0.1 of center]{\(p\)};
\draw[->] (site 1) -- (site 4);
\end{tikzpicture}
\caption{Definition of \(\partial_1\), \(\partial_2\)}
\label{jiaohuantu}
\end{wrapfigure}
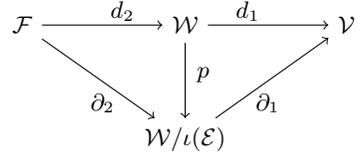
To construct a hypermap-homology quantum code, we need to choose bases for \(\mathcal{F}\), \(\mathcal{W}/\iota(\mathcal{E})\) and \(\mathcal{V}\). For \(\mathcal{F}\) and \(\mathcal{V}\), the faces and vertices are natural bases. To choose a basis for \(\mathcal{W}/\iota(\mathcal{E})\), firstly, we choose one \emph{special dart} for each edge \(e\), that is, a dart \(i\) that incident on \(e\), and denote \(S\)
for all these special darts, then it is straightforward to prove that the elements \(\omega_i+\iota(\mathcal{E})\) with \(i\in{B\setminus{S}}\) form a basis for \(\mathcal{W}/\iota(\mathcal{E})\). 

\subsection{Relation between hypermap codes and surface codes}
 In this section, we review pradeep's construction of equivalent surface codes from hypermap-homology quantum codes. 

To transform a hypermap-homology quantum code to a surface code, we firstly draw a (Pradeep's) curve (A smooth map from \([0,1]\) to the surface.) connecting vertices \(v_{\owns{i}}\) and \(v_{\owns{\alpha^{-1}(i)}}\) for each dart \(\omega_i\in{B}\), and label the curve `\(i\)' again! These curves should satisfy the following conditions:
\begin{itemize}
  \item None of them can be a single point.
  \item Curve \(i\) must lies entirely within the face to which dart \(\omega_i\) belongs excepting its end points \(v_{\owns{i}}\) and \(v_{\owns{\alpha^{-1}(i)}}\). 
  \item They cannot have any intersections or self-intersections  within a face.
\end{itemize}
Then for each special dart \(\omega_i\), we erase the curve \(i\). Now, the remaining curves together with the vertices of the original hypergraph form the 1-skeleton of a cell structure (CW-complex) on the surface \(\Sigma\). From this cell structure, we can make an surface code with underline chain complex the ordinary boundary maps \(\partial^*_2\) between the 2-cells and 1-cells, and \(\partial^*_1\) between the 1-cells and 0-cells. Algebraic topology shows that its homology group \(H^*_1=\ker \partial^*_{1}/\im \partial^*_{2}\) is just the singular homology group \(H_1(X)\) \cite{Hatcher}.\par
 Contrast of the two quantum codes can now be done step by step. For each face \(f_i\) of the original hypermap, we denote \(\omega_1^i\), \(\omega_2^i\), \(\omega_3^i\), ... , \(\omega_s^i\) the darts that belongs to it, and assume that \(\omega_{k_1}^i\), \(\omega_{k_2}^i\), \(\omega_{k_3}^i\), ... \(\omega_{k_r}^i\) (\(k_i\leqslant{s}\)) are special darts. Recall that all the equivalence classes \(\omega_i+\iota({\mathcal{E}})\) of the non-special darts \(\omega_i\in{B\setminus{S}}\) form a basis for \(\mathcal{W}/\iota(\mathcal{E})\), under this basis, we can write out the boundary of \(f_i\) as
 \begin{equation}
    \begin{aligned}
    \partial_2{f_i} & = \sum_{n=1}^{s}(\omega_{n}^{i}+\iota({\mathcal{E}}))\\
    & = \sum_{n\neq{k_j}, n\in\{1,\cdots,s\}}(\omega_{n}^{i}+\iota({\mathcal{E}}))\\& +\sum_{j=1}^{r}\sum_{
    \omega\in<\alpha^{-1}>\cdot{\omega_{k_j}^{i}}
    ,\omega\neq{\omega_{k_j}^{i}}
    }
 (\omega +\iota({\mathcal{E}}))
 \end{aligned}
\end{equation}
where \(<\alpha^{-1}>\cdot{\omega_{k_j}^{i}}\) means the orbit of \(\omega_{k_j}^{i}\) under the action of \(<\alpha^{-1}>\). All darts in the second row of equation (3) are non-special and therefore passes their labels to some of Pradeep's curves whose union forms the geometric boundary of a 2-cell \({f}^*_i\) in the new cell structure on \(\Sigma\) and the boundary matrix \([\partial_2^*f_i^*]\) under the  natural basis of Pradeep's curves are the same as the the matrix \([\partial_2f_i]\) of the hypermap boundary under the basis of non-special darts. Next, for each non-special dart \(\omega_i\), the boundary \(\partial_1(\omega_i+\iota(\mathcal{E}))\) are \(v_{\owns{i}}\) and \(v_{\owns{\alpha^{-1}(i)}}\), which is the same as the \(\partial^*_1\) boundary of its corresponding curve \(i\) ---Again, same matrix. By definition of homological quantum codes, these two codes are essentially the same code. Thus, We can reduce every hypermap homology code to an equivalent surface code.\par

    \section{Quantum codes of dual hypermaps }

 \subsection{Dual hypermap and Pradeep's surface code}
 There are concepts of dual hypermaps in both topological and combinatorial sense \cite{Martin}. In this section, we show that the structure of  Pradeep's surface code of the original hypermap will be more clear with the help of its dual hypermap.\par
 From a topological hypermap \(H=(\Sigma,\Gamma)\), we can construct the \emph{topological dual} \(H^*=(\Sigma_{op},\Gamma^*)\) as a hypermap satisfying the following conditions:
 \begin{itemize}
  \item \(\Sigma_{op}\) is the surface \(\Sigma\) with opposite orientation.
  \item The edges of \(H^*\) are the edges of \(H\).
  \item There is precisely one vertex of \(H^*\) for each face of \(H\), inside that face.
  \item For each dart labeled \(i\) of \(H\), there is precisely one dart of \(H^*\) that goes from \(e_{\owns{i}}\) to the vertex of \(H^*\) which lies inside the face \emph{\(f_{\owns{i}}\)} that \(i\) belongs to. We label this dart `\(i\)' again and all darts of \(H^*\) are therefore labeled.
  \item The tangent vector of dart \(i\) of \(H^*\) which starts at  \(e_{\owns{i}}\) must lie between tangent vectors of darts \(i\) and \(\alpha^{-1}(i)\) of \(H\) which start at the same edge.
  \item The darts of \(H^*\) must lie within faces excepting their end points \(e_{\owns{i}}\), and have no intersection or self-intersection inside the faces.
\end{itemize}

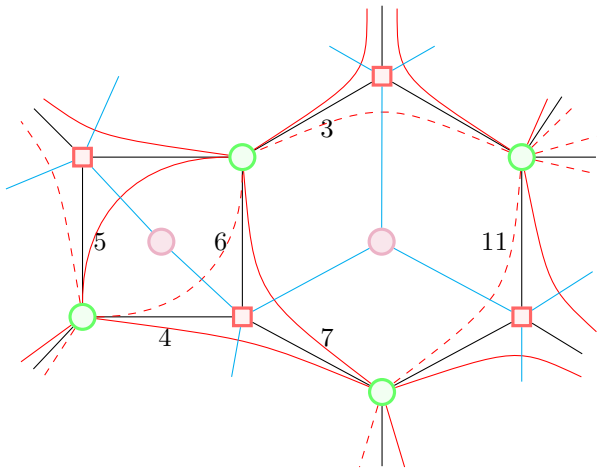
\begin{figure}[ht]
\centering
\begin{tikzpicture}[
roundnode/.style={circle, draw=green!60, fill=green!5, very thick, minimum size=2mm},
squarednode/.style={rectangle, draw=red!60, fill=red!5, very thick, minimum size=2mm},
fakenode/.style={rectangle, draw=orange!0, fill=blue!0, very thick, minimum size=2mm},
Roundnode/.style={circle, draw=purple!30, fill=purple!10, very thick, minimum size=1mm},
]
\node[squarednode]      (edge 1){} ;
\node[fakenode]        (site 1)       [below=0.8cm of edge 1]{};
\node[roundnode]        (vertex 3)       [below=2.8cm of site 1]{};
\node[roundnode]      (vertex 2)       [left=1.532cm of site 1] {};
\node[roundnode]      (vertex 1)       [right=1.532cm of site 1] {};
\node[squarednode]      (edge 3)       [below=1.8cm of vertex 2] {};
\node[squarednode]      (edge 2)       [below=1.8cm of vertex 1] {};
\node[roundnode]      (vertex 4)       [left=1.8cm of edge 3] {};
\node[squarednode]      (edge 4)       [left=1.8cm of vertex 2] {};
\node[fakenode]        (site 3)       [above=0.507cm of edge 4]{};
\node[fakenode]        (site 4)       [left=0.507cm of site 3 ]{};
\node[fakenode]        (site 6)       [below=0.507cm of vertex 4]{};
\node[fakenode]        (site 5)       [left=0.507cm of site 6]{};
\node[fakenode]        (site 2)       [above=0.8cm of edge 1]{};
\node[fakenode]        (site 7)       [below=0.8cm of vertex 3]{};
\node[fakenode]        (site 11)       [right=0.3cm of vertex 1]{};
\node[fakenode]        (site 12)       [right=0.8cm of vertex 1]{};
\node[fakenode]        (site 13)       [above=0.666cm of site 11]{};
\node[fakenode]        (site 9)       [right=0.666cm of edge 2]{};
\node[fakenode]        (site 8)       [below=0.3cm of site 9]{};
\node[fakenode]        (site 10)       [above=0.1cm of site 1]{};
\node[fakenode]        (site 14)       [left=0.366cm of site 10]{3};
\node[fakenode]        (site 15)       [below=2.27cm of site 14]{7};
\node[fakenode]        (site 16)       [below=0.8cm of vertex 2]{};
\node[fakenode]        (site 17)       [right=2.9cm of site 16]{11};
\node[fakenode]        (site 18)       [left=3.1cm of site 17]{6};
\node[fakenode]        (site 19)       [left=1.15cm of site 18]{5};
\node[fakenode]        (site 20)       [left=1.7cm of site 15]{4};
\node[Roundnode]        (Vertex 1)       [right=0.4cm of site 19]{};
\node[Roundnode]        (Vertex 2)       [left=1cm of site 17]{};

\draw[] (edge 1) -- (site 2);
\draw[] (edge 2) -- (site 8);
\draw[] (edge 4) -- (site 4);
\draw[] (vertex 1) -- (site 13);
\draw[] (vertex 1) -- (site 12);
\draw[] (vertex 3) -- (site 7);
\draw[] (vertex 4) -- (site 5);
\draw[] (vertex 1) -- (edge 2);
\draw[] (vertex 1) -- (edge 1);
\draw[] (vertex 2) -- (edge 1);
\draw[] (vertex 2) -- (edge 3);
\draw[] (vertex 3) -- (edge 3);
\draw[] (vertex 3) -- (edge 2);
\draw[] (vertex 2) -- (edge 4);
\draw[] (vertex 4) -- (edge 4);
\draw[] (vertex 4) -- (edge 3);
\draw[red!100,dashed] (vertex 1) .. controls (0,-0.3) .. (vertex 2);
\draw[red!100,dashed] (vertex 1) .. controls (1.7,-2.9) .. (vertex 3);
\draw[red!100] (vertex 3) .. controls (-1.7,-2.9) .. (vertex 2); 
\draw[red!100,dashed] (vertex 4.east) .. controls +(right:12mm) and +(down:12mm) .. (vertex 2.south);
\draw[red!100] (vertex 4.north) .. controls +(up:12mm) and +(left:12mm) .. (vertex 2.west);
\draw[red!100] (vertex 2) .. controls (-0.2,0) .. (-0.2,0.9);
\draw[red!100] (vertex 2) .. controls (-3.8,-0.8) .. (-4.5,-0.3);
\draw[red!100] (vertex 4) .. controls (-1.7,-3.5) .. (vertex 3);
\draw[red!100,dashed] (vertex 4) -- (-4.5,-4);
\draw[red!100,dashed] (vertex 3) -- (-0.3,-5.2);
\draw[red!100,dashed] (vertex 4) .. controls (-4.3,-1.3) .. (-4.8,-0.6);
\draw[red!100] (vertex 1) .. controls (0.2,0) .. (0.2,0.9);
\draw[red!100] (vertex 1) -- (2.2,-0.3);
\draw[red!100,dashed] (vertex 1) -- (2.55,-0.4);
\draw[red!100,dashed] (vertex 1) -- (2.85,-0.8);
\draw[red!100,dashed] (vertex 1) -- (2.85,-1.3);
\draw[red!100] (vertex 1) .. controls (2.2,-2.8) .. (2.85,-3.4);
\draw[red!100] (vertex 3) .. controls (1.8,-3.6) .. (2.65,-4);
\draw[red!100] (vertex 3)-- (0.3,-5.2);
\draw[red!100] (vertex 4) -- (-4.8,-3.8);

\draw[cyan!100] (Vertex 2) -- (edge 1);
\draw[cyan!100] (Vertex 2) -- (edge 3);
\draw[cyan!100] (Vertex 2) -- (edge 2);
\draw[cyan!100] (Vertex 1) -- (edge 4);
\draw[cyan!100] (Vertex 1) -- (edge 3);
\draw[cyan!100] (-2,-4) -- (edge 3);
\draw[cyan!100] (-3.5,0) -- (edge 4);
\draw[cyan!100] (-5,-1.5) -- (edge 4);
\draw[cyan!100] (-0.7,0.4) -- (edge 1);
\draw[cyan!100] (0.7,0.4) -- (edge 1);
\node[fakenode]        (site 22)       [below=2.8cm of vertex 1]{};
\draw[cyan!100] (site 22) -- (edge 2);
\draw[cyan!100] (2.8,-2.6) -- (edge 2);

\end{tikzpicture}
\caption{The dual hypermap and Pradeep's curves}
\label{fig:8}

\end{figure}
There are a few differences in above conditions from those in \cite{Martin}, these small modifications is used to deal situations like \(i=\sigma(i)\). See Appendices B and C for more on these singular situations. Figure \ref{fig:8} shows the picture of both dual hypermap \(H^*\) (with darts drawn in sky-blue and vertex in purple.) and Pradeep's curves of $H$ (in red), based on the original hypermap of Figure \ref{fig:1}. Here, darts 3, 11, 6 are in the set \(S\) of special darts, we draw their corresponding curves as dashed line which means that they will finally be erased. We can immediately see from Figure \ref{fig:8} that the  hypergraph of the dual hypermap and the graph (1-skeleton) of Pradeep's surface code (including those erased curves) are dual to each other in the sense of cell-complex. 
\begin{theorem}
\label{Pro 1}
We can construct Pradeep's surface code of the hypermap \(H=(\Sigma,\Gamma)\) by dualizing the graph \(\Gamma^*\) of the dual hypermap \(H^*=(\Sigma_{op},\Gamma^*)\) and erase the curves that intersect the darts of \(H^*\) whose label i belongs to \(S\).
\end{theorem}
\begin{proof}
 For any $i\in{B_n}$ , we deform the piecewise smooth curve which consists of dart $i$ and dart $\alpha^{-1}(i)$ of $H$ to a smooth curve inside the face $f_{\owns{i}}$, with its end points $v_i$ and $v_{\alpha^{-1}(i)}$ fixed. Now, due to the fifth conditions above that $H^*$ should satisfy, there is one and only one intersection of the curve with the dart $i$ of $H^*$ at the early stage of the deformation, and we stop here, then we has obtained the desired Pradeep's curves.
\end{proof}
As permutarions on the common label set $B_n$, the combinatorial hypermap corresponds to \(H^*\) are easily seen to be \((\alpha',\sigma')=(\alpha^{-1},\alpha^{-1}\sigma)\) and we have \(((\alpha')',(\sigma')')=(\alpha,\sigma)\), which suggests that in topological hypermap level, we may also have: 
\begin{equation}
    (H^*)^*=H
\end{equation}
Actually, the equals sign in equation (4) are correct in the sense of strong isomorphism mentioned in \cite{Martin}, i.e., we say topological hypermap \(H=(\Sigma,\Gamma)\) and \(H'=(\Sigma',\Gamma')\) are strongly isomorphic, and write \(H=H'\), if there exists an orientation-preserving homeomorphism \(u:\Sigma\rightarrow\Sigma'\) with \(u|_\Gamma\) giving an equality of hypergraphs.\par
In particular, \(H\) itself is one of the hypermap duals of \(H^*\) that satisfy the conditions of dual hypermap (See the proof of proposition 3.). Therefore, Pardeep's surface code of the dual hypermap \(H^*\) with the same set of special darts \(S\) can also be constructed easily by dualizing the original graph \(\Gamma\) and erase those curves which crosses the special darts.

\subsection{Edge hypermap codes, $\triangle$-duals and contrary maps}
In the last section, we have related the hypermap codes of $H$ to the dual hypermap $H^*$ by giving a simple geometric relation between the Pradeep's surface code and the underline bipartite graph of $H^*$. However, it is still not clear weather there are simple relations between the codes themselves, i.e, between the pair of hypermap quantum codes $(\mathcal{C},\mathcal{C^*})$ from $(H,H^*)$. We try to answer the question in this section. 

Recall that for a topological hypermap, we have four vector spaces \(\mathcal{V}\), \(\mathcal{E}\), \(\mathcal{F}\), \(\mathcal{W}\), from which we have the chain complex \(\mathcal{F}\xrightarrow{d_2}\mathcal{W}\xrightarrow{d_1}\mathcal{V}\), further more, we constructed the subspace \(\iota(\mathcal{E})\), together with which we have the chain complex \(\mathcal{F}\xrightarrow{\partial_2}\mathcal{W}/\iota(\mathcal{E})\xrightarrow{\partial_1}\mathcal{V}\) by Figure \ref{jiaohuantu}. Then by choosing a special dart for each edge to form a basis \(\omega_i+\iota(\mathcal{E})\) for \(W/\iota(\mathcal{E})\) with \(i\in{B}\setminus{S}\), the hypermap homology code can be constructed. Now, in the first chain complex, we replace \(\mathcal{F}\) by \(\mathcal{E}\), and \(d_2\) by \(\iota:e\mapsto\sum_{i\in{e}} \omega_i\), but leave \(d_1\) unchanged, then we still have chain complex \(d_1\circ{\iota}=0\). Conversely, we replace \(\mathcal{E}\) by \(\mathcal{F}\), and also replace \(\iota\) by 
\(d_2\), which gives us the subspace \(d_2(\mathcal{F})\) and the quotient \(\mathcal{W}/d_2(\mathcal{F})\). Because \(d_1\circ{d_2}=0\), again the map \(\hat{\partial_1}:\omega_i+d_2(\mathcal{F})\mapsto{d_1\omega_i}\) is well defined and we have the diagram as is shown in 
\begin{wrapfigure}{r}{0.4\textwidth}
\centering
\begin{tikzpicture}[
squarednode/.style={rectangle, draw=black!0, fill=green!0, very thick, minimum size=2mm},
]
\node[squarednode]      (site 1)        {\(\mathcal{W}\)};
\node[squarednode]      (site 2)       [right=1.6 of site 1] {\(\mathcal{V}\)};
\node[squarednode]      (site 3)       [left=1.6 of site 1] {\(\mathcal{E}\)};
\node[squarednode]      (site 4)       [below=1 of site 1] {\(\mathcal{W}/d_2(\mathcal{F})\)};
\node[squarednode]      (fake l)       [left=0.3 of site 4]{};
\node[squarednode]      (fake r)       [right=0.3 of site 4]{};
\node[squarednode]      (p1)       [above=0.1 of fake l]{\(\hat{\partial_2}\)};
\draw[->] (site 3) -- (site 4);
\node[squarednode]      (p2)       [above=0.1 of fake r]{\(\hat{\partial_1}\)};
\draw[->] (site 4) -- (site 2);
\node[squarednode]      (fake l+)       [left=0.4 of site 1]{};
\node[squarednode]      (fake r+)       [right=0.4 of site 1]{};
\node[squarednode]      (d1)       [above=-0.2 of fake l+]{\(\iota\)};
\node[squarednode]      (d2)       [above=-0.2 of fake r+]{\(d_1\)};
\draw[->] (site 1) -- (site 2);
\draw[->] (site 3) -- (site 1);
\node[squarednode]      (center)       [below=0.3 of site 1]{};
\node[squarednode]      (p)       [right=-0.1 of center]{\(p\)};
\draw[->] (site 1) -- (site 4);
\end{tikzpicture}
\caption{Definition of \(\hat{\partial_1}\), \(\hat{\partial_2}\)}
\label{fig:9}
\end{wrapfigure}
Figure \ref{fig:9}. Simply speaking, we interchange \(\mathcal{E}\) and \(\mathcal{F}\), \(\iota\) and \(d_2\) in Figure \ref{jiaohuantu} to obtain a new chain complex \(\hat{\partial_1}\circ\hat{\partial_2}=0\). To construct stabilizer code, we choose one special dart for each face (A dart that belongs to this face) to obtain \(\hat{S}\), then a basis for \(\mathcal{W}/d_2(\mathcal{F})\) is \(\omega_i+d_2(\mathcal{F})\) with \(i\in{B\setminus{\hat{S}}}\). In this article we call homological code obtained this way the \emph{edge hypermap-homology quantum code} to distinguish it from the original one proposed by Martin which we call \emph{face hypermap homology quantum code} from now on. Pradeep's curves can be added under the same conditions of section 2.3, which again form equivalent surface codes.

Also, from a given topological hypermap \(H=(\Sigma,\Gamma)\), besides its dual, we can construct an another hypermap \(H^{\triangle}=(\Sigma_{op},\Gamma^\triangle)\) satisfying the following constraints:
\begin{itemize}
  \item \(\Sigma_{op}\) is the surface \(\Sigma\) with opposite orientation.
  \item The vertices of \(H^\triangle\) are the vertices of \(H\).
  \item There is precisely one edge of \(H^\triangle\) for each face of \(H\), inside that face.
  \item For each dart labeled \(i\) of \(H\), there is precisely one dart of \(H^\triangle\) that goes from \(v_{\owns{i}}\) to the edge of \(H^\triangle\) which lies inside the face \emph{\(f_{\owns{i}}\)} that \(i\) belongs to. We label this dart `\(i\)' again and all darts of \(H^\triangle\) are therefore labeled.
  \item The tangent vector of dart \(i\) of \(H^\triangle\) which starts at  \(v_{\owns{i}}\) must lie between tangent vectors of dart \(i\) and \(\sigma^{-1}(i)\) of \(H\) which start at the same vertex.
  \item The darts \(i\) of \(H^\triangle\) must lie within faces \(f_{\owns{i}}\) excepting the points \(v_{\owns{i}}\), and have no intersection or self-intersection inside the faces.
\end{itemize}
As with the dual hypermap \(H^*\), in the sense of strong isomorphism, \(H^\triangle\) is also unique, satisfying  \((H^\triangle)^\triangle=H\) and therefore could be regarded as another kind of dual hypermap. Further more, by definition of \(H^\triangle\), it's combinatorial hypermap is \((\hat{\alpha},\hat{\sigma})=(\sigma^{-1}\alpha,\sigma^{-1})\), which means that the faces of \(H^\triangle\) are the orbits of 
\begin{equation}
<\hat{\alpha}^{-1}\hat{\sigma}>=<(\sigma^{-1}\alpha)^{-1}\sigma^{-1}>=<\alpha^{-1}>
    \end{equation}
, i.e., the edges of \(H\), while the edges of \(H^\triangle\) are the orbits of
\begin{equation}
 <\hat{\alpha}^{-1}>=<\alpha^{-1}\sigma>
\end{equation}
, i.e., the faces of \(H\). Thus, when choosing a set of special darts for a face hypermap code of \(H\), their corresponding darts in \(\Gamma^\triangle\) form exactly a set of special darts for an edge hypermap code of \(H^\triangle\), and vice versa. With this set of common special darts, the face hypermap code $\mathcal{C}$ of \(H\) equals the edge hypermap code $\mathcal{C}^\triangle$ of \(H^\triangle\) due to equation (5) and (6). This equivalence can also be seen by the fact that there exists a common set of Pradeep's curves for these two codes:
\begin{theorem}
\label{Pro 2}
With a common set of special darts \(S=\hat{S}\), the topological hypermap \(H\) and \(H^\triangle\) can have a common set of Pradeep's curves. As a consequence, the face hypermap-homology quantum code of \(H\) equals the edge hypermap-homology quantum code of \(H^\triangle\) under \(S\).
\end{theorem}
\begin{proof}
For a dart \(\omega_i\) of \(H\), there is an edge $E_{\owns{i}}$ of $H^\triangle$ lies inside the face $f_{\owns{i}}$ to which $\omega_i$ belongs, also, there is a dart $\Omega_i$ of $H^\triangle$ which coresspond to $i$ and connects $E_{\owns{i}}$ and $v_{\owns{i}}$. Because $\omega_i$ and $\Omega_i$ have a common vertex that can be denoted by $v_{\owns{i}}$ and we have $v_{\owns{\alpha^{-1}(i)}}=v_{\owns{\hat{\alpha}^{-1}\hat{\sigma}(i)}}=v_{\owns{\hat{\alpha}^{-1}(i)}}$, we conclude that the vertices of \(\alpha^{-1}\omega_i\) and $\hat{\alpha}^{-1}\Omega_i$ are equal. Therefore,
the darts $\omega_i$, $\alpha^{-1}\omega_i$, $\Omega_i$, $\hat{\alpha}^{-1}\Omega_i$ enclose a tetragon (might be a degenerated one) whose interior is contained in $f_{\owns{i}}$ and any curve on $\Sigma$ which connects the points \(v_{\owns{i}}\) and $v_{\owns{\alpha^{-1}(i)}}$ and lies totally in the interior of the tetragon excepting the ends must satisfies the constraints of the Pradeep's curve for both \(H\) and $H^\triangle$. However, when \(i\) is in $S$, we do not draw any curve inside the tetragon.
\end{proof}

Next, we will construct both $H^*$ and $H^\triangle$ simultaneously by forcing the edges of $H^\triangle$ coinside with the vertices of $H^*$ at all faces of \(H\), and avoiding any intersection of the the darts of both $H^*$ and $H^\triangle$ aside from their common ends. Remember that for \(H\), we write the label of a dart inside the face that counterclockwise incident to the dart with respect to its edge, in this article, this rule should also be applied to $H^*$ and $H^\triangle$. However, because the orientations of both $H^*$ and $H^\triangle$ are opposite to that of \(H\), we can meet this rule for all the three hypermaps by writing a single \(i\) inside the triangle that enclosed by darts $\omega_i$, $\omega_i^*$ (this means the dart of \(H^*\) that corresponds to $i$.) and $\Omega_i$, which is shown in Figure \ref{fig:10}. This triangle has been mentioned in literature  (see \cite{Triangle} or Definition 4.7 in \cite{Martin}) without explicitly defining \(H^\triangle\).

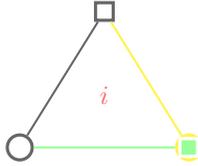
\begin{figure}[ht]
\centering
\begin{tikzpicture}[
squarednode/.style={rectangle, draw=black!60, fill=cyan!0, very thick, minimum size=2mm},
roundnode/.style={circle, draw=yellow!80, fill=green!0, very thick, minimum size=2mm},
fakenode/.style={rectangle, very thick, minimum size=2mm},
capnode1/.style={rectangle, draw=orange!0, fill=green!40, very thick, minimum size=2mm},
capnode2/.style={circle, draw=black!60, fill=purple!0, very thick, minimum size=2mm},
capnode3/.style={rectangle, draw=blue!0, fill=blue!0, very thick, minimum size=2mm},
]
\node[squarednode]      (edge 1)       [] {};
\node[fakenode]      (site 1)       [below=1.532 of edge 1] {};
\node[roundnode]      (vertex 2)       [right=0.8 of site 1] {};
\node[capnode2]      (vertex 1)       [left=0.8 of site 1] {};
\node[capnode1]      (edge 2)       [right=0.845 of site 1] {};
\draw[thick, black!60] (edge 1) -- (vertex 1);
\draw[thick, yellow!80] (edge 1) -- (vertex 2);
\draw[thick, green!40] (vertex 2) -- (vertex 1);
\node[fakenode]      (site 2)       [above=0.3 of site 1] {\textcolor{red!60}{\( i\)}};
\node[fakenode]      (site 3)       [left=0.5 of vertex 1] {};

\end{tikzpicture}
\caption{Darts $\omega_i$, $\omega^*_i$, $\Omega_i$ with their colors black, orange and green.} 
\label{fig:10}
\end{figure}

Finally, let's define the \emph{contrary map} \(C\) of a topological hypermap $H=(\Sigma,\Gamma)$ as \(C=(\Sigma_{op}, \overline{\Gamma})\),  where $\overline{\Gamma}$ represents the hypergraph obtained by interchanging the edges and vertices of \(\Gamma\). Thus, if we interchange the vertices and edges of the  hypergraph of \(H^\triangle\), and turn the normal field back, we get the contrary map \(H^\nabla\). Fortunately, this time, Figure \ref{fig:10} is still the correct way to set the label. The combinatorial hypermap of $H^\nabla$ is simply
\begin{equation}
    (\overline{\alpha},\overline{\sigma})=(\hat{\sigma},\hat{\alpha})
\end{equation}
by definition, thus we have
\begin{equation}
    <\overline{\alpha}^{-1}>=<\hat{\sigma}^{-1}>=<\sigma>=<\alpha'^{-1}\sigma'>,
\end{equation}
and
\begin{equation}
    <\overline{\alpha}^{-1}\overline{\sigma}>=<\hat{\sigma}^{-1}\hat{\alpha}>=<\alpha>=<\alpha'^{-1}>.
\end{equation}

Equations (8) and (9) show that the faces and edges of \(H^*\) and \(H^\nabla\) are interchanged, which we were already familiar with in the discussion of \(H\) and \(H^\triangle\). Actually, we have
\begin{theorem}
The contary map \(H^\nabla\) is a \(\triangle\)-dual of the dual hypermap $H^*$, i.e, $H^\nabla=(H^*)^\triangle$.
\end{theorem}
\begin{proof}
 The darts of all the four hypermaps $(H,H^*,H^\triangle,H^\nabla)$ are  naturally related by the same label $i\in{B}$. We denote the 4 darts labelel $i$ by $(\omega_i,\omega_i^*,\Omega_i,\Omega^*_i)$ with $\Omega_i=\Omega^*_i$ due to the definition of contrary map and \(\omega_i^*\), \(\Omega_i.\) are defined directly from $\omega_i$. Now, we fix an dart \(\omega^*_i\). The face of \(H^*\) to which \(\omega^*_i\) belongs is the orbit $<\alpha'^{-1}\sigma'>\cdot\omega^*_i$, we denote it $c^*_i$. The dart of $H^*$ which also belongs to the closure of $c^*_i$ and share the same edge with $\omega_i^*$ is $\alpha'^{-1}\omega_i^*=\omega^*_{\alpha'^{-1}(i)}=\omega^*_{\alpha(i)}$. By definition of dual hypermap, the tangent vector of $\omega_i^*$ at its edges $e_{\owns{i}}$ must lie between dart $\omega_i$ and dart $\omega_{\alpha^{-1}(i)}$ of $H$, the tangent vector of $\omega^*_{\alpha(i)}$ at its edges $e_{\owns{\alpha(i)}}=e_{\owns{i}}$ must lie between dart $\omega_{\alpha(i)}$ and dart $\omega_{\alpha\alpha^{-1}(i)}=\omega_i$ of $H$, geometrically, this indicates that the tangent vector of $\omega_i$ at its edges $e_{\owns{i}}$  lies between dart $\omega^*_i$ and dart $\omega^*_{\alpha(i)}=\omega^*_{\alpha'^{-1}(i)}$ of $H^*$ and further more, that the vertex $v_{\owns{i}}$ of \(\omega_i\) lies inside $c^*_i$ because $\Gamma$ and $\Gamma^*$ do not intersect beyond their common edges. Since edges of $H$ must lie on $\Gamma^*$ and therefore can not be inside $c^*_i$,  then, for each vertex $v$ of $H$ inside $c^*_i$, there is an dart $\omega$ of $H$ that connects $v$ and an edge $e$ of $H$ on the boundary of $c^*_i$. Notice that for each dart \(\omega_\epsilon\) of $H$, its corresponding dart \(\omega^*_\epsilon\)  in $H^*$ is the first dart of $H^*$ that \(\omega_\epsilon\)  will meet when it is turned counterclockwise around \(e_{\owns{\epsilon}}\) (At least in the vicinity of \(e_{\owns{\epsilon}}\).), there must be an inner dart of $c^*_i$ that correspond to $\omega$, say $(\alpha'^{-1}\sigma')^k(\omega^*_i)$, we have \(\omega=\omega_{(\alpha'^{-1}\sigma')^k(i)}=\omega_{\sigma^k(i)}\), and therefore \(v=v_{\owns{i}}\), which tells us that $v_{\owns{i}}$ is the only vertex of $H$ that lies inside $c^*_i$ (Here, we have actually given a more rigorous proof of equation (4)). Now, because the vertices of $H$ are exactly the edges of $H^\nabla$, we have that for each face of $H^*$, there is only one edge of $H^\nabla$ inside. By definition of $H^\triangle$, there is a dart \(\Omega_i=\Omega^*_i\) that connects $v_{\owns{i}}$ and the vertex $v^*_{\owns{i}}$ of \(\omega^*_i\)  (see Figure \ref{fig:10}). More over, the tangent vector of \(\Omega^*_i\) at \(v^*_{\owns{i}}\)  lies between $\omega_i^*$ and $\omega^*_{\sigma'^{-1}(i)}$. To see this, notice that the dart $\omega^*_{\sigma'^{-1}(i)}$ corresponds the dart $\omega_{(\alpha^{-1}\sigma)^{-1}(i)}$ of \(H\) whose edge lies counterclockwise to $v_{\owns{i}}$ at the boundary of the face $c_i$ that \(\omega_i\) belongs to, while the edge of the dart $\omega^*_i$ is the edge of dart \(\omega_i\) and lie clockwise to $v_{\owns{i}}$ at the same boundary. Finally, with the fact that $H^\nabla$ has opposite orientation to $H^*$, we have proved our result. 
\end{proof}
When choosing a set of special darts $\omega^*_i$ with  \(i\in S'\subset{B}\) for \(H^*\), the darts $\Omega^*_i$ ($i\in S'$) are those which correspond to $\omega^*_i$ in the sense of $\triangle$-dual by the proof of proposition 3. Now, by proposition 2 and 3, we have that the edge hypermap code $\mathcal{C}^\nabla$ of $H^\nabla$ under \(S'\) equals the face hypermap code $\mathcal{C}^*$ of \(H^*\) under \(S'\) (This can also be proved directly using equation (8) and (9).), and furthermore, they have a common surface code of Pradeep's.\par

\begin{figure}[ht]
\centering
\begin{tikzpicture}[
roundnode/.style={circle, draw=green!60, fill=green!5, very thick, minimum size=2mm},
squarednode/.style={rectangle, draw=red!60, fill=red!5, very thick, minimum size=2mm},
fakenode/.style={rectangle, draw=orange!0, fill=blue!0, very thick, minimum size=2mm},
Roundnode/.style={circle, draw=purple!30, fill=purple!10, very thick, minimum size=1mm},
Fakenode/.style={circle, draw=orange!0, fill=orange!0, very thick, minimum size=2mm},
]
\node[squarednode]      (edge 1){} ;
\node[fakenode]        (site 1)       [below=0.8cm of edge 1]{};
\node[roundnode]        (vertex 3)       [below=2.8cm of site 1]{};
\node[roundnode]      (vertex 2)       [left=1.532cm of site 1] {};
\node[roundnode]      (vertex 1)       [right=1.532cm of site 1] {};
\node[squarednode]      (edge 3)       [below=1.8cm of vertex 2] {};
\node[squarednode]      (edge 2)       [below=1.8cm of vertex 1] {};
\node[roundnode]      (vertex 4)       [left=1.8cm of edge 3] {};
\node[squarednode]      (edge 4)       [left=1.8cm of vertex 2] {};
\node[fakenode]        (site 3)       [above=0.507cm of edge 4]{};
\node[fakenode]        (site 4)       [left=0.507cm of site 3 ]{};
\node[fakenode]        (site 6)       [below=0.507cm of vertex 4]{};
\node[fakenode]        (site 5)       [left=0.507cm of site 6]{};
\node[fakenode]        (site 2)       [above=0.8cm of edge 1]{};
\node[fakenode]        (site 7)       [below=0.8cm of vertex 3]{};
\node[fakenode]        (site 11)       [right=0.3cm of vertex 1]{};
\node[fakenode]        (site 12)       [right=0.8cm of vertex 1]{};
\node[fakenode]        (site 13)       [above=0.666cm of site 11]{};
\node[fakenode]        (site 9)       [right=0.666cm of edge 2]{};
\node[fakenode]        (site 8)       [below=0.3cm of site 9]{};
\node[fakenode]        (site 10)       [above=0.1cm of site 1]{};
\node[fakenode]        (site 14)       [right=0.8cm of vertex 2]{3};
\node[fakenode]        (site 15)       [below=2.27cm of site 14]{};
\node[fakenode]        (site 16)       [below=0.8cm of vertex 2]{};
\node[fakenode]        (site 17)       [right=2.9cm of site 16]{};
\node[fakenode]        (site 18)       [left=3.1cm of site 17]{6};
\node[fakenode]        (site 19)       [left=1.15cm of site 18]{5};
\node[fakenode]        (site 20)       [left=1.7cm of site 15]{4};
\node[Roundnode]        (Vertex 1)       [right=0.4cm of site 19]{};
\node[Roundnode]        (Vertex 2)       [left=1cm of site 17]{};

\draw[] (edge 1) -- (site 2);
\draw[] (edge 2) -- (site 8);
\draw[] (edge 4) -- (site 4);
\draw[] (vertex 1) -- (site 13);
\draw[] (vertex 1) -- (site 12);
\draw[] (vertex 3) -- (site 7);
\draw[] (vertex 4) -- (site 5);
\draw[] (vertex 1) -- (edge 2);
\draw[] (vertex 1) -- (edge 1);
\draw[] (vertex 2) -- (edge 1);
\draw[] (vertex 2) -- (edge 3);
\draw[] (vertex 3) -- (edge 3);
\draw[] (vertex 3) -- (edge 2);
\draw[] (vertex 2) -- (edge 4);
\draw[] (vertex 4) -- (edge 4);
\draw[] (vertex 4) -- (edge 3);
\draw[red!100,dashed] (vertex 1) .. controls (0,-0.3) .. (vertex 2);
\draw[red!100,dashed] (vertex 1) .. controls (1.7,-2.9) .. (vertex 3);
\draw[red!100,dashed] (vertex 4.east) .. controls +(right:12mm) and +(down:12mm) .. (vertex 2.south);
\draw[red!100] (vertex 4.north) .. controls +(up:12mm) and +(left:12mm) .. (vertex 2.west);
\draw[red!100] (vertex 2) .. controls (-0.2,0) .. (-0.2,0.9);
\draw[red!100] (vertex 2) .. controls (-3.8,-0.8) .. (-4.5,-0.3);
\draw[red!100] (vertex 4) .. controls (-1.7,-3.5) .. (vertex 3);
\draw[red!100,dashed] (vertex 4) -- (-4.5,-4);
\draw[red!100,dashed] (vertex 3) -- (-0.3,-5.2);
\draw[red!100,dashed] (vertex 4) .. controls (-4.3,-1.3) .. (-4.8,-0.6);
\draw[red!100] (vertex 1) .. controls (0.2,0) .. (0.2,0.9);
\draw[red!100] (vertex 1) -- (2.2,-0.3);
\draw[red!100,dashed] (vertex 1) -- (2.55,-0.4);
\draw[red!100,dashed] (vertex 1) -- (2.85,-0.8);
\draw[red!100,dashed] (vertex 1) -- (2.85,-1.3);
\draw[red!100] (vertex 1) .. controls (2.2,-2.8) .. (2.85,-3.4);
\draw[red!100] (vertex 3) .. controls (1.8,-3.6) .. (2.65,-4);
\draw[red!100] (vertex 3)-- (0.3,-5.2);
\draw[red!100] (vertex 4) -- (-4.8,-3.8);

\draw[cyan!100] (Vertex 2) -- (edge 1);
\draw[cyan!100] (Vertex 2) -- (edge 3);
\draw[cyan!100] (Vertex 2) -- (edge 2);
\draw[cyan!100] (Vertex 1) -- (edge 4);
\draw[cyan!100] (Vertex 1) -- (edge 3);
\draw[cyan!100] (-2,-4) -- (edge 3);
\draw[cyan!100] (-3.5,0) -- (edge 4);
\draw[cyan!100] (-5,-1.5) -- (edge 4);
\draw[cyan!100] (-0.7,0.4) -- (edge 1);
\draw[cyan!100] (0.7,0.4) -- (edge 1);
\node[fakenode]        (site 22)       [below=2.8cm of vertex 1]{};
\draw[cyan!100] (site 22) -- (edge 2);
\draw[cyan!100] (2.8,-2.6) -- (edge 2);
\draw[green!30,thick] (Vertex 1) -- (vertex 4);
\draw[green!30,thick] (Vertex 1) -- (vertex 2);
\draw[green!30,thick] (Vertex 2) -- (vertex 1);
\draw[green!30,thick] (Vertex 2) -- (vertex 2);
\draw[green!30,thick] (Vertex 2) -- (vertex 3);
\node[Fakenode]        (dot 1)       [above=0.8cm of vertex 2]{};
\node[Fakenode]        (dot 2)       [below right=0.5cm of vertex 3]{};
\node[Fakenode]        (dot 3)       [below left=0.5cm of vertex 3]{};
\node[Fakenode]        (dot 4)       [below right=0.5cm and 0.2cm of vertex 4]{};
\node[Fakenode]        (dot 5)       [above left=0.05cm and 0.5cm of vertex 4]{};
\node[Fakenode]        (dot 6)       [above left=0.5cm and -0.2cm of vertex 1]{};
\node[Fakenode]        (dot 7)       [below right=0.2cm and 0.5cm of vertex 1]{};
\node[Fakenode]        (dot 8)       [right=0.6cm of Vertex 2]{11};
\node[Fakenode]        (dot 9)       [right=0.7cm of edge 3]{7};
\draw[green!30,thick] (dot 1) -- (vertex 2);
\draw[green!30,thick] (dot 6) -- (vertex 1);
\draw[green!30,thick] (dot 7) -- (vertex 1);
\draw[green!30,thick] (dot 3) -- (vertex 3);
\draw[green!30,thick] (dot 2) -- (vertex 3);
\draw[green!30,thick] (dot 5) -- (vertex 4);
\draw[green!30,thick] (dot 4) -- (vertex 4);
\draw[red!100] (vertex 3) .. controls (-1.7,-2.9) .. (vertex 2);
\draw[green!30,thick] (3,-0.5) -- (vertex 1);
\draw[blue!100] (Vertex 1) .. controls (-3.8,-1) .. (-3.4,0);
\draw[blue!100,dashed] (-3.8,0) .. controls (-4.1,-0.9) .. (-5,-1.3);
\draw[blue!100] (Vertex 1) .. controls (-4,-1.3) .. (-5,-1.7);
\draw[blue!100,dashed] (Vertex 2) .. controls (-1.8,-3) .. (Vertex 1);
\draw[blue!100] (Vertex 1) .. controls (-2.2,-3.2) .. (-2.2,-3.9);
\draw[blue!100] (Vertex 2) .. controls (-1.7,-3.3) .. (-1.8,-4);
\draw[blue!100,dashed] (Vertex 2) .. controls (-0.2,-0.2) .. (-0.8,0.2);
\draw[blue!100] (Vertex 2) .. controls (0.2,-0.2) .. (0.8,0.2);
\draw[blue!100] (-0.5,0.7) .. controls (0,0.5) .. (0.5,0.7);
\draw[blue!100,dashed] (Vertex 2) .. controls (1.9,-2.7) .. (2.7,-2.3);
\draw[blue!100] (Vertex 2) .. controls (1.6,-3.4) .. (1.7,-4);
\draw[blue!100] (2.9,-2.9) .. controls (2.3,-3.4) .. (2.2,-4);

\end{tikzpicture}
\caption{\(H,H^*,H^\triangle,H^\nabla\) and their surface codes. The darts of $H^{\triangle(\nabla)}$ are represented by the green lines, the Pradeep's curves of \(\mathcal{C}^{*(\nabla)}\) are represented by the blue curves.}
\label{fig:11}
\end{figure}
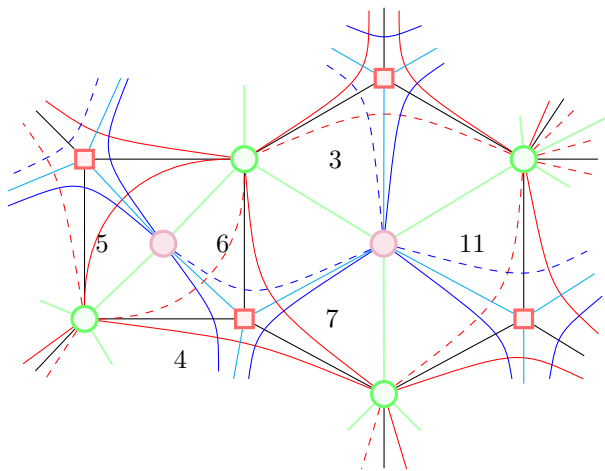

As a reference, in Figure \ref{fig:11}, we put all four hypermaps \(H,H^*,H^\triangle,H^\nabla\) with their common Pradeep's surface codes in to a single local picture of Figure \ref{fig:1}. Here, we do not explicitly draw out the edges or vertices of $H^\triangle$ and $H^\nabla$, but use round nodes in different colors to represent the two sets of sites of the underline bipartite graph, with only the green lines left which represent their common darts. Also, we use the blue (red) curves to represent a common set of Pradeep's curves of $\mathcal{C}^*$ ($\mathcal{C}$) and $\mathcal{C}^\nabla$ ($\mathcal{C}^\triangle$).\par
Now, if all four sets of special darts are corresponded in the natural way, then we can study the relations between the pair $(\mathcal{C},\mathcal{C}^*)$ by studying those of the equivalent pair $(\mathcal{C}^\triangle,\mathcal{C}^\nabla)$, and we see immediately that $\mathcal{C}^\triangle$ and $\mathcal{C}^\nabla$ have their underling topological hypermaps simply the contrary maps of each others, in addition, their special darts are geometrically coincide ($\Omega_i^{*}=\Omega_i$). Intuitively, from figure \ref{fig:11}, we can see that within a face of $H^{\triangle(\nabla)}$, the Pradeep's curves of the pair are dual to each other in the sense that they share opposite sets of bipartite nodes, actually, this last simple phenomena is the primary observation of this article.

\bibliographystyle{unsrt}
\bibliography{hypermap}

\appendix
\appendixpage
\addappheadtotoc
\section{Construction of a Topological hypermap from a conbinatorial one}
Notice that when the permutation \(\alpha^{-1}\sigma\) continuously acts on a dart, the dart will circle around it's face, i.e, the orbit of the dart under the action of subgroup \(<\alpha^{-1}\sigma>\) is the face to which the dart belongs. Actually we can construct an oriented topological hypermap by finding the orbits of the subgroup \(<\alpha^{-1}\sigma>\) of a combinatorial hypermap \((\alpha, \sigma)\) and then gluing these `faces' together.

Let's suppose that we have found all orbits of \(<\alpha^{-1}\sigma>\), they form a partition \(I\) of  \(B_n=\{1,2,...,n\}\). When all subsets of the partition is acted by \(\alpha^{-1}\), they form another partition \(O\), which, in the topological hypermaps case, represents the `outer' darts of the faces (Recall Figure 2.). Now, every element of \(B_n\) belongs exactly to one subset of \(I\) and one subset of \(O\). For each orbit, which consists of elements \(\{i_1, i_2,...,i_s\}\) with \(i_{k+1}=\alpha^{-1}\sigma(i_k)\) for any \(i_k\) with \(k<s\) and \(i_{1}=\alpha^{-1}\sigma(i_s)\) , we draw a face (a polygon) on the paper in the way as Figure 2 shows, which has s inner edges label by \(i_k\) and s outer edges label by \(\alpha^{-1}(i_k)\) . Each element of \(B_n\) labels an inner edge of precisely one of these faces ( because for any of these faces, its inner edges are exactly one subset of partion \(I\) ) and a outer edge of precisely one of these faces (because for any one of these faces, its outer edges are exactly one subset of partion \(O\)), and we can glue the corresponding faces along the two edges labeled by this element. Now, for all element of \(B_n\), we glue faces this way. (There is possibility that a face be glued with itself, which happens when \(i_k=\alpha^{-1}(i_j)\) for some \(k\), \(j\in \{1,2,...,s\}\).) All faces have their naturally pointed normal fields, i.e., normal vectors pointing out from paper, this helps us matching the orientations when we glue an inner edge with an outer edge. All we need to do is to glue squared node with squared node, round node with round node, then their naturally pointed normal fields will be matched together and form a global normal field. To see this, notice that when walking along an edge from a round node to a squared node, we are actually clockwise circling the normal vectors of the face when its an inner edge, otherwise, we would circle the normal vectors counterclockwise. When edges are being glued the right way, i.e., round node to round node, squared node to squared node,  we will see the local picture right before gluing from some suitable direction as Figure \ref{fig:3} shows in which the left-side edge is a outer edge, while the right-side one, an inner edge. Now, according to the previous observation, both of the  normal fields of face \(f_1\) and \(f_2\) shoud be pointing again from paper to us (at least, within the local picture), which means that their naturally normal vector fields are perfectly matched along the edge we glued. Denote \(\Sigma\) for the space after gluing, the neighbors of every points of \(\Sigma\) are now surface like, i.e., homeomorphic to \(\mathbf{R}^2\), excepting those that come from nodes (squared or round) of the original faces, we  call them sites.
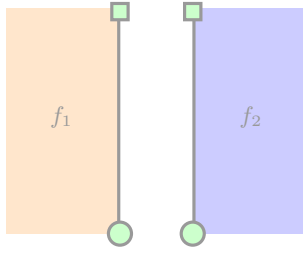
\begin{wrapfigure}{l}{0.35\textwidth}
\begin{tikzpicture}[
squarednode/.style={rectangle, draw=black!40, fill=green!20, very thick, minimum size=2mm},
roundnode/.style={circle, draw=black!40, fill=green!20, very thick, minimum size=2mm},
fakenode/.style={rectangle, very thick, minimum size=2mm},
capnode1/.style={rectangle, draw=blue!20, fill=blue!20, very thick, minimum size=2mm},
capnode2/.style={rectangle, draw=orange!20, fill=orange!20, very thick, minimum size=2mm},
capnode3/.style={rectangle, draw=blue!0, fill=blue!0, very thick, minimum size=2mm},
]
\fill[orange!20] (-0.5,0) rectangle (-2,3);
\fill[blue!20] (0.5,0) rectangle (2,3);

\node[fakenode]      (site 1)        {};
\draw[very thick, black!40] (0.5,0) -- (0.5,3);
\node[roundnode]      (vertex 1)       [right=0.18cm of site 1] {};
\node[squarednode]      (edge 1)       [above=2.65 of vertex 1] {};
\draw[very thick, black!40] (-0.5,0) -- (-0.5,3);
\node[roundnode]      (vertex 2)       [left=0.18cm of site 1] {};
\node[squarednode]      (edge 2)       [above=2.65cm of vertex 2] {};
\node[fakenode]      (site 2)       [right=1cm of site 1] {};
\node[fakenode]      (site 3)       [left=1cm of site 1] {};
\node[capnode1]      (site 4)       [above=1.15cm of site 2]{\textcolor{black!40}{\(f_2\)}};
\node[capnode2]      (site 5)       [above=1.15cm of site 3]{\textcolor{black!40}{\(f_1\)}};
\end{tikzpicture}
\caption{A local picture right before gluing} 
\label{fig:3}
\end{wrapfigure}
In the vicinity of these sites, each dart (glued edges) has two and only two faces that incident on it, therefore, the neighbor of these sites must be a single cone (It's impossible to have two or more cones with only their apex glued into a single site.) without self intersection, which is topologically a disk. So, \(\Sigma\) is surface like near every points---we've got an oriented surface! The edges and nodes of the original faces becomes the darts and sites of an hypermap on \(\Sigma\), and the interiors of the faces themselves become the faces of the hypermap. We've transformed a combinatorial hypermap to a topological one, and it's easy to check that this topological hypermap with normal field the one we glued are exactly the original hypermap \((\alpha, \sigma)\) in the sense of permutations on \(B_n\), and the transformations between these two kinds of hypermaps are mutually reversible. Further more, every two darts of the newly constructed hypermap can be mutually transformed to each other by circling around the sites for finite many times due to the transitivity of \(<\sigma, \alpha>\) , which indicates that the surface is connected. Also, the orbits of \(<\alpha^{-1}\sigma>\) are finite, which indicates that the surface is compact.\par
 As an example, let's construct a hypermap-homology quantum code of the pair \((\alpha, \sigma)\) with \(\alpha=\)(4 3 2 1)(5 7 8 6), \(\sigma=\)(7 1 6 3)(5 2 8 4) in the permutation group \(S_8\). The hypermap has 2 edges, \(e_1=\)(4 3 2 1), \(e_2=\)(5 7 8 6), and 2
 vertices \(v_1=\)(7 1 6 3), \(v_2=\)(5 2 8 4). The faces can be calculated as follow:
 \begin{equation}\label{faces}
    \begin{aligned}
    \alpha^{-1}\sigma & = 
    \begin{pmatrix}
    1 & 2 & 3 & 4 & 5 & 6 & 7 & 8\\
    8 & 7 & 5 & 6 & 3 & 4 & 2 & 1
    \end{pmatrix}\\
                      & = (1\quad8)(2\quad7)(3\quad5)(4\quad6)
    \end{aligned}
    \end{equation}
    from equation \ref{faces}, we draw 4 faces as Figure \ref{fig:5} shows, 
    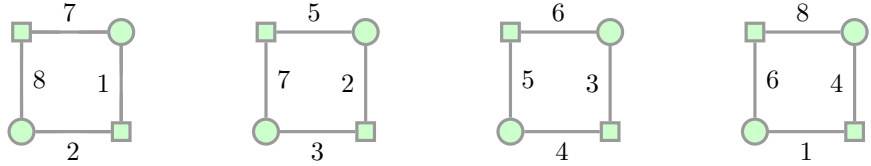
\begin{figure}[h]
        \centering
        \begin{tikzpicture}[
        squarednode/.style={rectangle, draw=black!40, fill=green!20, very thick, minimum size=2mm},
        roundnode/.style={circle, draw=black!40, fill=green!20, very thick, minimum size=2mm},
        fakenode/.style={rectangle, draw=orange!0, fill=blue!0, minimum size=1mm},
        ]
        \node[roundnode]      (vertex 1)       [] {};
        \node[squarednode]      (edge 2)       [right=1cm of vertex 1] {};
        \node[roundnode]      (vertex 2)       [right=1.6cm of edge 2] {};
        \node[squarednode]      (edge 4)       [right=1cm of vertex 2] {};
        \node[roundnode]      (vertex 3)       [right=1.6cm of edge 4] {};
         \node[squarednode]      (edge 6)       [right=1cm of vertex 3] {};
          \node[roundnode]      (vertex 4)       [right=1.6cm of edge 6] {};
          \node[squarednode]      (edge 8)       [right=1cm of vertex 4] {};
          \node[squarednode]      (edge 1)       [above=1cm of vertex 1] {};
          \node[squarednode]      (edge 3)       [above=1cm of vertex 2] {};
          \node[squarednode]      (edge 5)       [above=1cm of vertex 3] {};
          \node[squarednode]      (edge 7)       [above=1cm of vertex 4] {};
          \node[roundnode]      (vertex 8)       [above=1cm of edge 2] {};
           \node[roundnode]      (vertex 7)       [above=1cm of edge 4] {};
            \node[roundnode]      (vertex 6)       [above=1cm of edge 6] {};
             \node[roundnode]      (vertex 5)       [above=1cm of edge 8] {};
             \node[fakenode]      (site 1)       [right=0.375cm of vertex 1] {}; 
             \node[fakenode]      (site 2)       [below=-0.1cm of site 1] {2};
             \node[fakenode]      (site 3)       [right=0.375cm of edge 1] {};
             \node[fakenode]      (site 4)       [above=-0.1cm of site 3] {7};
             \draw [very thick, black!40](vertex 8) -- (edge 1);
             \draw [very thick, black!40](vertex 1) -- (edge 2);
             \node[fakenode]      (site 5)       [below=0.375cm of edge 1] {}; 
             \node[fakenode]      (site 6)       [right=-0.1cm of site 5] {8}; 
             \draw [very thick, black!40](vertex 1) -- (edge 1);
             \node[fakenode]      (site 7)       [below=0.375cm of vertex 8] {}; 
             \node[fakenode]      (site 8)       [left=-0.1cm of site 7] {1}; 
             \draw [very thick, black!40](vertex 8) -- (edge 2);
             
             \node[fakenode]      (site 9)       [right=0.375cm of vertex 1] {}; 
             \node[fakenode]      (site )       [below=-0.1cm of site 1] {2};
             \node[fakenode]      (site 3)       [right=0.375cm of edge 1] {};
             \node[fakenode]      (site 4)       [above=-0.1cm of site 3] {7};
             \draw [very thick, black!40](vertex 8) -- (edge 1);
             \draw [very thick, black!40](vertex 1) -- (edge 2);
             \node[fakenode]      (site 5)       [below=0.375cm of edge 1] {}; 
             \node[fakenode]      (site 6)       [right=-0.1cm of site 5] {8}; 
             \draw [very thick, black!40](vertex 1) -- (edge 1);
             \node[fakenode]      (site 7)       [below=0.375cm of vertex 8] {}; 
             \node[fakenode]      (site 8)       [left=-0.1cm of site 7] {1}; 
             \draw [very thick, black!40](vertex 8) -- (edge 2);
             
             \node[fakenode]      (site 1')       [right=0.375cm of vertex 2] {}; 
             \node[fakenode]      (site 2')       [below=-0.1cm of site 1'] {3};
             \node[fakenode]      (site 3')       [right=0.375cm of edge 3] {};
             \node[fakenode]      (site 4')       [above=-0.1cm of site 3'] {5};
             \draw [very thick, black!40](vertex 7) -- (edge 3);
             \draw [very thick, black!40](vertex 2) -- (edge 4);
             \node[fakenode]      (site 5')       [below=0.375cm of edge 3] {}; 
             \node[fakenode]      (site 6')       [right=-0.1cm of site 5'] {7}; 
             \draw [very thick, black!40](vertex 2) -- (edge 3);
             \node[fakenode]      (site 7')       [below=0.375cm of vertex 7] {}; 
             \node[fakenode]      (site 8')       [left=-0.1cm of site 7'] {2}; 
             \draw [very thick, black!40](vertex 7) -- (edge 4);
            
             \node[fakenode]      (site 1'')       [right=0.375cm of vertex 3] {}; 
             \node[fakenode]      (site 2'')       [below=-0.1cm of site 1''] {4};
             \node[fakenode]      (site 3'')       [right=0.375cm of edge 5] {};
             \node[fakenode]      (site 4'')       [above=-0.1cm of site 3''] {6};
             \draw [very thick, black!40](vertex 6) -- (edge 5);
             \draw [very thick, black!40](vertex 3) -- (edge 6);
             \node[fakenode]      (site 5'')       [below=0.375cm of edge 5] {}; 
             \node[fakenode]      (site 6'')       [right=-0.1cm of site 5''] {5}; 
             \draw [very thick, black!40](vertex 3) -- (edge 5);
             \node[fakenode]      (site 7'')       [below=0.375cm of vertex 6] {}; 
             \node[fakenode]      (site 8'')       [left=-0.1cm of site 7''] {3}; 
             \draw [very thick, black!40](vertex 6) -- (edge 6);
             
              \node[fakenode]      (site 1''')       [right=0.375cm of vertex 4] {}; 
             \node[fakenode]      (site 2''')       [below=-0.1cm of site 1'''] {1};
             \node[fakenode]      (site 3''')       [right=0.375cm of edge 7] {};
             \node[fakenode]      (site 4''')       [above=-0.1cm of site 3'''] {8};
             \draw [very thick, black!40](vertex 5) -- (edge 7);
             \draw [very thick, black!40](vertex 4) -- (edge 8);
             \node[fakenode]      (site 5''')       [below=0.375cm of edge 7] {}; 
             \node[fakenode]      (site 6''')       [right=-0.1cm of site 5'''] {6}; 
             \draw [very thick, black!40](vertex 4) -- (edge 7);
             \node[fakenode]      (site 7''')       [below=0.375cm of vertex 5] {}; 
             \node[fakenode]      (site 8''')       [left=-0.1cm of site 7'''] {4}; 
             \draw [very thick, black!40](vertex 5) -- (edge 8);

        \end{tikzpicture}
        \caption{Faces of \((\alpha, \sigma)\)  }
        \label{fig:5}
    \end{figure}
    
    \noindent After labeling their inner edges, the outer edges can also be labeled using \(\alpha^{-1}\) as was shown in Figure 2. Now we can glue the four faces by gluing (Of course, squared node to squared node, round node to round node.) an inner edge with a outer edge that has the same label, and we have got a topological hypermap. To construct a quantum code, let us choose dart 2 and 5 as special darts, then a basis for \(\mathcal{W}/\iota(\mathcal{E})\) is (1, 3, 4, 6, 7, 8) with \(i\) represent the equivalence class \(\omega_{i}+\iota(\mathcal{E})\). Also, a basis for \(\mathcal{F}\) is (\(f_1\), \(f_2\), \(f_3\), \(f_4\)), which represents the faces from left to right in Figure \ref{fig:5}, and a basis for \(\mathcal{V}\) is (\(v_1\), \(v_2\)). Under these bases, the binary check matrix \(A\) (see 2.1) is calculated as:
    \begin{equation}\label{b mat}
    \begin{aligned}
    H_X & =
    \begin{pmatrix}
    1 & 1 & 1 & 1 & 1 & 1 \\
    1 & 1 & 1 & 1 & 1 & 1 
    \end{pmatrix}
    \qquad
    H_Z & =
    \begin{pmatrix}
    1 & 0 & 0 & 0 & 0 & 1 \\
    1 & 1 & 1 & 0 & 1 & 0 \\
    0 & 1 & 0 & 1 & 1 & 1 \\
    0 & 0 & 1 & 1 & 0 & 0 
    \end{pmatrix}
     \end{aligned}
    \end{equation}
    
\noindent from matrix \(A\), we can directly write out the generators for the hypermap-homology code: 
\begin{equation}
\begin{aligned}
X_{v_1}=X_{v_2} & =X_1X_3X_4X_6X_7X_8\\
Z_{f_1} & =Z_1Z_8\\
Z_{f_2} & =Z_1Z_3Z_4Z_7\\
Z_{f_3} & =Z_3Z_6Z_7Z_8\\
Z_{f_4} & =Z_4Z_6
\end{aligned}
\end{equation}
only 4 of the above 6 operators are independent, say \(X_{v_1}\), \(Z_{f_1}\), \(Z_{f_2}\), \(Z_{f_3}\), which can be seem from equation \ref{b mat}. This tells us that we have \(6-4=2\) logic qubits---\(k=\dim{H_1}=2\). To physically realize a logic qubit, we can firstly attach a qubit \(|\phi_i\rangle\) to each non-special dart \(i\in{B\setminus{S}}\), which makes the above operators local, then project the whole system \(|\psi\rangle=\bigotimes_{i\in{B\setminus{S}}}|\phi_i\rangle\) to the stabilizer code space \(\mathcal{C}\) by measuring the four operators and then applying some extra gates according to the results.

\section{Degenerated cases}
Our description about the construction of the  Pradeep's curves is slightly different from that in \cite{Pradeep}, which makes it more straightforward to deal with singular cases like the two shown in Figure \ref{fig:sing} (Dashed line represents a erased curve which however, we are interested in.).
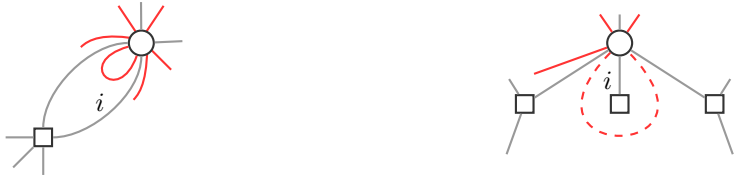
\begin{figure}[h]
    \centering
    \begin{tikzpicture}[
    squarednode/.style={rectangle, draw=black!80, fill=green!0, thick, minimum size=2mm},
    roundnode/.style={circle, draw=black!80, fill=green!0, thick,  minimum size=2mm},
    fakenode/.style={rectangle, draw=orange!0, fill=blue!0, minimum size=1mm},
        ]
    \draw[red!80,thick] (1.3,1.3) .. controls (0.8,1.4) and (0.6,1.3) .. (0.5,1.2);   
    \draw[red!80,thick] (1.38,1.38) .. controls (1.4,0.8) and (1.3,0.6) .. (1.2,0.5); 
    \draw[red!80,thick] (1.3,1.3) -- (1.0,1.75);
    \draw[red!80,thick] (1.3,1.3) -- (1.7,0.9); 
    \draw[red!80,thick] (1.3,1.3) -- (1.6,1.75);
    \node[squarednode]    (edge 1) {};
     \node[fakenode] (site 2) [right=0.5cm of edge 1] {};
    \node[fakenode] (site 3) [above=0.1cm of site 2] {\(i\)};
    \node[fakenode]   (site 1) [above=1cm of edge 1] {};
    \node[roundnode]    (vertex 1) [right=1cm of site 1 ] {};
    \draw[black!40,thick] (vertex 1.south) .. controls +(down:5mm) and +(right:5mm) .. (edge 1.east);
    \draw [black!40,thick](vertex 1.west) .. controls +(left:5mm) and +(up:5mm) .. (edge 1.north);
    \draw[red!80,thick] (vertex 1) .. controls (0.4,1) and (1,0.4) .. (vertex 1);
    \draw[black!40,thick] (vertex 1) -- (1.3,1.8);
    \draw[black!40,thick] (vertex 1) -- (1.85,1.28);
    \draw [black!40,thick](edge 1) -- (-0.5,0);
    \draw [black!40,thick](edge 1) -- (-0.4,-0.4);
    \draw[black!40,thick] (edge 1) -- (0,-0.5);
     \node[fakenode] (site 14) [right=7.25cm of edge 1] {};
    \node[fakenode] (site 15) [above=0.4cm of site 14] {\(i\)};
     \node[roundnode]    (vertex 2) [right=6cm of vertex 1 ] {};
    \node[squarednode]   (edge 2) [below=0.5cm of vertex 2] {};
    \node[squarednode]   (edge 3) [left=1cm of edge 2] {};
    \node[squarednode]   (edge 4) [right=1cm of edge 2] {};
    \node[fakenode]   (site 4) [below=1.3cm of vertex 2] {};
    \node[fakenode]   (site 5) [left=1.3cm of site 4] {};
    \node[fakenode]   (site 6) [right=1.3cm of site 4] {};
     \draw [black!40,thick](edge 2) -- (vertex 2);
     \draw [black!40,thick](edge 3) -- (vertex 2);
     \draw [black!40,thick](edge 4) -- (vertex 2);
     \draw [black!40,thick](edge 3) -- (site 5);
     \draw [black!40,thick](edge 4) -- (site 6);
      \draw[red!80, thick, dashed] (vertex 2) .. controls (site 5) and (site 6) .. (vertex 2);
      \node[fakenode]   (site 7) [above=1cm of site 5] {};
      \node[fakenode]   (site 8) [above=1cm of site 6] {};
      \draw [black!40,thick](edge 3) -- (site 7);
      \draw [black!40,thick](edge 4) -- (site 8);
      \node[fakenode]   (site 9) [above=0.2cm of vertex 2] {};
       \draw [black!40,thick](vertex 2) -- (site 9);
       \node[fakenode]   (site 10) [left=0.1cm of site 9] {};
      \node[fakenode]   (site 11) [right=0.1cm of site 9] {}; 
       \draw[red!80,thick] (vertex 2) -- (site 10);
        \draw[red!80,thick] (vertex 2) -- (site 11);
        \node[fakenode]   (site 12) [above=0.1cm of edge 3] {};
        \node[fakenode]   (site 13) [above=0.1cm of edge 4] {};
        \draw[red!80,thick] (vertex 2) -- (site 12);
    
    \end{tikzpicture}
    \caption{Cases when \(v_{\owns{i}}=v_{\owns{\alpha^{-1}(i)}}\). }
    \label{fig:sing}
\end{figure}
The left part of Figure \ref{fig:sing} is a local picture of the case \(v_{\owns{i}}=v_{\owns{\alpha^{-1}(i)}}\) but \(i\neq{\alpha^{-1}(i)}\). Here, the red line represents the Pradeep's curves, and the explicitly labeled dart is non-special. This dart passes its label \(i\)  to the red self-circle inside the face to which it belongs. For cases where \(i={\alpha^{-1}(i)}\), \(\omega_i\) is the only dart that incident on the edge \(e_{\owns{i}}\), and is therefore an special dart. Thus, no matter how we draw the curve \(i\), it is to be erased, as is shown in the right part of Figure \ref{fig:sing} (For edge hypermap codes, it is the self-circle of the left part figure that should be erased.  ).\par
\section{An ambiguity}
For a specific example of Pradeep's construction, see Figure 8 and 9 of his paper \cite{Pradeep}. To us, his example contains a subtle situation in which our slightly modified procedure may encompass some ambiguities, i.e., there are two ways of drawing curves in the situation shown as follow:\par
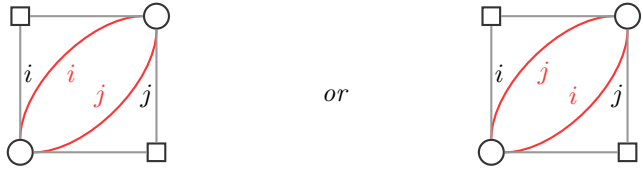
\begin{figure}[ht]
    \centering
    \begin{tikzpicture}[
    squarednode/.style={rectangle, draw=black!80, fill=green!0, thick, minimum size=2mm},
    roundnode/.style={circle, draw=black!80, fill=green!0, thick,  minimum size=2mm},
    fakenode/.style={rectangle, draw=orange!0, fill=blue!0, minimum size=1mm},
        ]
        
        \node[squarednode]    (edge 1) {};
        \node[fakenode]    (site 1) [below=0.5cm of edge 1 ] {};
         \node[fakenode]    (site 2) [right=-0.2cm of site 1 ] {\(i\)};
          \node[fakenode]    (site 3) [right=0.2cm of site 2 ]     {\textcolor{red!80}{\(i\)}};
          \node[roundnode]    (vertex 1) [right=1.5cm of edge 1 ] {};
        \node[roundnode]    (vertex 2) [below=1.5cm of edge 1 ] {};
        \node[squarednode]    (edge 2) [below=1.5cm of vertex 1 ] {};
         \node[fakenode]    (site 1') [above=0.5cm of edge 2 ] {};
         \node[fakenode]    (site 2') [left=-0.2cm of site 1' ]
         {\(j\)};
         \node[fakenode]    (site 3') [left=0.2cm of site 2' ]     {\textcolor{red!80}{\(j\)}};
        \draw[red!80,thick] (vertex 1.south) .. controls +(down:7mm) and +(right:7mm) .. (vertex 2.east);
         \draw[red!80,thick] (vertex 1.west) .. controls +(left:7mm) and +(up:7mm) .. (vertex 2.north);
        \draw [black!40,thick](edge 2) -- (vertex 2);
        \draw [black!40,thick](edge 2) -- (vertex 1);
        \draw [black!40,thick](vertex 1) -- (edge 1);
        \draw [black!40,thick](vertex 2) -- (edge 1);
        
         \node[fakenode]    (site or) [right=2cm of site 2' ] {\emph{or}};
        
         \node[squarednode]    (edge 1) [right=6cm of edge 1 ] {};
        \node[fakenode]    (site 1) [below=0.5cm of edge 1 ] {};
         \node[fakenode]    (site 2) [right=-0.2cm of site 1 ] {\(i\)};
          \node[fakenode]    (site 3) [right=0.2cm of site 2 ]     {\textcolor{red!80}{\(j\)}};
          \node[roundnode]    (vertex 1) [right=1.5cm of edge 1 ] {};
        \node[roundnode]    (vertex 2) [below=1.5cm of edge 1 ] {};
        \node[squarednode]    (edge 2) [below=1.5cm of vertex 1 ] {};
         \node[fakenode]    (site 1') [above=0.5cm of edge 2 ] {};
         \node[fakenode]    (site 2') [left=-0.2cm of site 1' ]
         {\(j\)};
         \node[fakenode]    (site 3') [left=0.2cm of site 2' ]     {\textcolor{red!80}{\(i\)}};
        \draw[red!80,thick] (vertex 1.south) .. controls +(down:7mm) and +(right:7mm) .. (vertex 2.east);
         \draw[red!80,thick] (vertex 1.west) .. controls +(left:7mm) and +(up:7mm) .. (vertex 2.north);
        \draw [black!40,thick](edge 2) -- (vertex 2);
        \draw [black!40,thick](edge 2) -- (vertex 1);
        \draw [black!40,thick](vertex 1) -- (edge 1);
        \draw [black!40,thick](vertex 2) -- (edge 1);

    \end{tikzpicture}
    \caption{ Ambiguity occurs when \(d_1\omega_i=d_1\omega_j\), but \(i\neq{j}\).}
    \label{fig:7}
\end{figure}
 \noindent In the left part of Figure \ref{fig:7}, curve \(i\) is closer to dart \(i\), while in the right part, curve \(j\) is closer to dart \(i\). Although both left and right part of Figure \ref{fig:7} is the correct way of adding curves according to our 3 conditions proposed in the main article, the right side one is different from Pradeep's original example. This would not cause any problem, for the domain enclosed by curve \(i\) and \(j\) is a 2-cell, and does not affect the whole topological structure no matter which one of the two curves is erased when it's necessary, further more, the existence of the 2-cell \(f^*_i\) in the main article is also not affected (\(f^*_i\) can be seem as a flower that is formed by attaching the petal-cells enclosed by curves of the darts \(\omega\in<\alpha^{-1}>\cdot{\omega_{k_j}^{i}}
    \)  for each special dart \(\omega_{k_j}^{i}\) to the main disc enclosed by curves of all the inner darts \(\omega_{n}^{i}\).  These petal-cells  must have no intersection, which is not affected by the different ways of Figure \ref{fig:7}.).

\medskip

\end{document}